\RequirePackage{amsmath} 
\documentclass[runningheads, envcountsame, a4paper]{llncs}
\usepackage[utf8]{inputenc}
\usepackage{amsmath}
\usepackage{amssymb}
\usepackage{mathrsfs}
\usepackage{hyperref}
\usepackage{graphicx}
\usepackage{mathtools}
\usepackage{verbatim}
\usepackage{adjustbox}
\usepackage{graphicx}
\usepackage{bm}
\usepackage{tikz}
\usetikzlibrary{automata}
\usetikzlibrary{arrows,positioning}
\usetikzlibrary{calc}
\usepackage{tikzscale}
\usetikzlibrary{shapes.geometric}
\usetikzlibrary{arrows,automata}
\usetikzlibrary{matrix}
\usetikzlibrary{decorations}
\usetikzlibrary{decorations.pathreplacing}
\usepackage{stmaryrd}
\usepackage{microtype}
\usepackage{bbm}
\usepackage{braket}
\usepackage{array}
\usepackage{cleveref}
\usepackage[normalem]{ulem} 

\usepackage{extarrows} 

\newcommand{\powerset}{\raisebox{.15\baselineskip}{\large\ensuremath{\wp}}}

\DeclareMathAlphabet{\mathpzc}{OT1}{pzc}{m}{it}

\makeatletter
\@addtoreset{section}{partnumber}
\makeatother

\newcommand{\suchthat}{\;\ifnum\currentgrouptype=16 \middle\fi|\;}

\newcommand{\mrk}[1]{{#1}'}
\newcommand{\oldstack}[3]{%
{\ifthenelse{\equal{#1}{1}}{%
\mrk{#2}
}%
{#2}}_{#3}%
}
\newcommand{\stack}[3]{%
[%
{\ifthenelse{\equal{#1}{1}}{%
\mrk{#2}
}%
{#2}}\ {#3}%
]%
}
\newcommand{\tstack}[2]{%
[#1,\ #2]%
}

\newcommand{\va}[1]{\stackrel{#1}{\longrightarrow}}
\newcommand{\ourpath}[1]{\stackrel{#1}{\leadsto}}
\newcommand{\flush}[1]{\stackrel{#1}{\Longrightarrow}}

\newcommand{\shift}[1]{\stackrel{#1}{\dashrightarrow}}

\newcommand{\nont}[4]{\langle^{#1} #2, #3 {}^{#4} \rangle}
\newcommand{\chain}[3]{{}^{#1}\!\left[ #2 \right]\!{}^{#3}}

\newcommand{\tconfig}[3]{\langle #1, \ #2, \ #3 \rangle}

\newcommand{\symb}[1]{\mathop{symbol}(#1)}

\newcommand{\state}[1]{\mathop{state}(#1)}

\newcommand{\parrow}{}
\newcommand{\iarrow}{dashed}

\newcommand{\pair}[2]{\langle {#1},{#2}\rangle} 

\newcommand{\comp}[1]{ \stackrel {{#1}} \vdash }

\newcommand{\lp}{\llparenthesis}
\newcommand{\rp}{\rrparenthesis}

\newcommand{\opl}{OPL}
\newcommand{\opa}{OPA}
\newcommand{\opg}{OPG}

\newcommand{\tagxml}[2]{%
\langle %
{\ifthenelse{\equal{#1}{1}}{%
/#2  
}%
{#2}}%
\rangle%
}

\makeatother

\tikzset{
	path/.style={dotted},
	every edge/.style={draw,solid},
	normal/.style={solid},
	siblings/.style={dashed},
	support/.style={decorate, decoration={snake,amplitude=.4mm,segment length=2mm,post length=1mm}}
}

\newcommand{\pref}{\mathbb{P}}

\begin{document}

\title{Cyclic Operator Precedence Grammars for Parallel Parsing}

\author{Michele Chiari\inst{3}, Dino Mandrioli\inst{1}, Matteo Pradella\inst{1,2}}
	\institute{
		Dipartimento di Elettronica, Informazione e Bioingegneria (DEIB), Politecnico di Milano, Piazza Leonardo Da Vinci 32, 20133 Milano, Italy
    \and
        IEIIT, Consiglio Nazionale delle Ricerche, via Ponzio 34/5, 20133 Milano, Italy
		\email{\{dino.mandrioli,matteo.pradella\}@polimi.it}
    \and
        Institute of Computer Engineering, TU Wien, Treitlstraße 3, 1040 Vienna, Austria
        \email{michele.chiari@tuwien.ac.at}
	}
\authorrunning{M. Chiari, D. Mandrioli, M. Pradella}

\maketitle 
\begin{abstract}
Operator precedence languages (OPL) enjoy the local parsability property, 
which essentially means that a code fragment enclosed within a pair of markers 
---playing the role of parentheses---
can be compiled with no knowledge of its external context.
Such a property has been exploited to build parallel compilers for languages formalized as OPLs.
It has been observed, however, that when the syntax trees of the sentences have a linear substructure, 
its parsing must necessarily proceed sequentially making it impossible to split such a subtree into chunks to be processed in parallel.
Such an inconvenience is due to the fact that so far much literature on OPLs has assumed the hypothesis that equality precedence relation cannot be cyclic. 
This hypothesis was motivated by the need to keep the mathematical notation as simple as possible.

We present an enriched version of operator precedence grammars, called cyclic, that allows to use a simplified version of regular expressions in the right hand sides of grammar's rules;
for this class of operator precedence grammars the acyclicity hypothesis of the equality precedence relation is no more needed to guarantee the algebraic properties of the generated languages. 
The expressive power of the cyclic grammars is now fully equivalent to that of other formalisms defining OPLs such as operator precedence automata, monadic second order logic and operator precedence expressions.
As a result cyclic operator precedence grammars now produce also unranked syntax trees and sentences with flat unbounded substructures that can be naturally partitioned into chunks suitable for parallel parsing.

\end{abstract}

\section{Introduction}
\emph{Operator precedence languages (OPL)} are a ``historical'' family of languages invented by R. Floyd \cite{Floyd1963} to support fast deterministic parsing. Together with their  \emph{operator precedence grammars (OPG)}, they are still used within modern compilers to parse expressions with operators ranked by  priority.
The key feature that makes them well amenable for efficient parsing and compilation is that the syntax tree  of a sentence is determined
exclusively by three binary precedence relations over the terminal alphabet that are easily pre-computed from the grammar productions. 
For readers unacquainted with
 OPLs, we anticipate an example:  the arithmetic sentence  $a + b \times c$ does not make manifest the natural structure
$(a + (b \times c))$,  but the latter is implied by the fact that the plus operator yields precedence to the times.

Early theoretical investigation \cite{Crespi-ReghizziMM1978}, originally motivated by grammar inference goals, realized that, thanks to the  tree structure assigned
to  strings by the precedence relations, many closure
properties of regular languages hold for OPLs too, but only after a long intermission, renewed research \cite{Crespi-ReghizziM12} proved further algebraic properties of OPLs. 
Then, in \cite{LonatiEtAl2015} the  \emph{Operator
Precedence automata (OPA)} were defined as a formalism, paired with OPGs as well as pushdown automata are paired with context-free grammars (CFG), to formalize the recognition and parsing of OPLs.
In the same paper a \emph{monadic second-order (MSO) logic} characterization of OPLs
that  naturally extends  the classic one for regular languages was also produced. 
Furthermore \cite{MPC20} introduced the \emph{operator precedence expressions (OPE)} which extend traditional regular expressions in the same way as the MSO logic for OPLs extends the traditional one for regular languages. 
The extension of fundamental algebraic and logic properties of regular languages to OPLs has been completed in \cite{Henzinger23} where a characterization of OPLs in terms of a syntactic congruence with a finite number of equivalence classes is given.

Another distinguishing property  of OPLs is their \emph{local parsability}, 
i.e. the fact that a chunk of text included within a pair of symmetric precedence relations can be deterministically parsed even without knowing its context
\cite{BarenghiEtAl2012a}. This feature was exploited to produce a parallel parser generator which exhibited high performances w.r.t. traditional sequential parsers \cite{BarenghiEtAl2015}.

A different research branch ---which however is not the core aspect of this paper--- exploited the algebraic and logic properties of this family to extend to it the successful 
verification techniques based on model checking, with comparable performances with those obtained for regular languages \cite{DBLP:conf/cav/ChiariMP20}. In summary, this historical family of languages enjoys properties that allowed relevant modern applications such as automatic verification and parallel compilation.

It must be pointed out, however, that many ---not all--- of the algebraic properties discovered during such a research activity that lasted several decades were proved by assuming a hypothesis on the precedence relations defined on the input alphabet.
Although from a theoretical point of view this hypothesis slightly affects the generative power of OPGs ---but is not necessary, e.g., for OPAs and MSO logic so that these two formalisms are a little more powerful than OPGs---, 
so far no practical limitation due to it was discovered in terms of formalizing the syntax of real-life programming and data description languages. 
Thus, it has been constantly adopted in the various developments to avoid making the mathematical notation and technical details too cumbersome.

The recent contribution \cite{DBLP:conf/hpcasia/LiT23}, however, pointed out a weakness of the parallel compilation algorithm described in \cite{BarenghiEtAl2015} which in some cases hampers partitioning the input to be parsed in well-balanced chunks, so that the benefits of parallelism are affected.
Intuitively, the weakness is due to the fact that the ``normal'' precedence relation on the arithmetic operator $+$ compels to parse a sequence thereof by associating them either to the left or to right so that parsing becomes necessarily sequential in this case.
The authors also proposed a special technique to overtake this difficulty by allowing for an acceptable level of ambiguity, which in the case of OPGs determines a conflict for some precedence relations on the terminal alphabet.

Such normal precedence relation, however, which either let $+$ yield precedence to, or take precedence over itself, has no correspondence with the semantics of the operation whose result does not depend on the order of association.
So, why not giving the various occurrences of the $+$ operator the same precedence level as suggested by arithmetic's laws?%
\footnote{Of course, assuming that the use of parentheses does not alter the normal precedence between the various operators.}
Figure \ref{fig:ass vs flat plus} gives an intuitive idea of the different structures given to a sequence of $+$ operators by traditional OPGs and the natural semantics of the sum operation. More precise explanations will be given in the following sections.

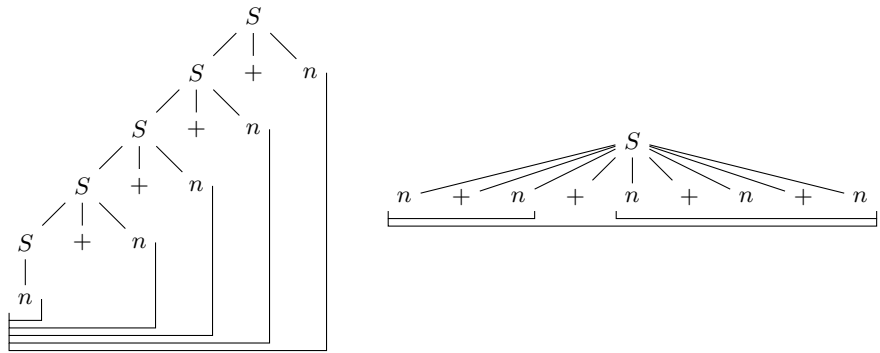
\begin{figure}[tb]
\centering
\begin{minipage}{.4\linewidth}
\begin{tikzpicture}[scale=0.5]
\node{$S$}
child{ node{$S$}
    child{ node{$S$}
        child{ node{$S$}
            child{ node{$S$}
                child{ node (n1) {$n$} }
            }
            child{ node{$+$} }
            child{ node (n2) {$n$} }
        }
        child{ node{$+$} }
        child{ node (n3) {$n$} }
    }
    child{ node{$+$} }
    child{ node (n4) {$n$} }
}
child{ node{$+$} }
child{ node (n5) {$n$} };
\coordinate (bot) at (n1.south);
\coordinate (dwn) at (0,-.2);
\draw let \p1=(n1.west), \p2=(n1.east), \p3=(bot) in (\x1,\y3) -- ($(\x1,\y3) + (dwn)$) -- ($(\x2,\y3) + (dwn)$) -- (\x2,\y2);
\draw let \p1=(n1.west), \p2=(n2.east), \p3=($(bot) + (dwn)$) in (\x1,\y3) -- ($(\x1,\y3) + (dwn)$) -- ($(\x2,\y3) + (dwn)$) -- (\x2,\y2);
\draw let \p1=(n1.west), \p2=(n3.east), \p3=($(bot) + 2*(dwn)$) in (\x1,\y3) -- ($(\x1,\y3) + (dwn)$) -- ($(\x2,\y3) + (dwn)$) -- (\x2,\y2);
\draw let \p1=(n1.west), \p2=(n4.east), \p3=($(bot) + 3*(dwn)$) in (\x1,\y3) -- ($(\x1,\y3) + (dwn)$) -- ($(\x2,\y3) + (dwn)$) -- (\x2,\y2);
\draw let \p1=(n1.west), \p2=(n5.east), \p3=($(bot) + 4*(dwn)$) in (\x1,\y3) -- ($(\x1,\y3) + (dwn)$) -- ($(\x2,\y3) + (dwn)$) -- (\x2,\y2);
\end{tikzpicture}
\end{minipage}
\begin{minipage}{.5\linewidth}
\begin{tikzpicture}[scale=0.5]
\node{$S$}
child{ node (n1) {$n$} }
child{ node{$+$} }
child{ node (n2) {$n$} }
child{ node{$+$} }
child{ node (n3) {$n$} }
child{ node{$+$} }
child{ node (n4) {$n$} }
child{ node{$+$} }
child{ node (n5) {$n$} };
\coordinate (bot) at (n1.south);
\coordinate (dwn) at (0,-.2);
\draw let \p1=(n1.west), \p2=(n2.east), \p3=(bot) in (\x1,\y3) -- ($(\x1,\y3) + (dwn)$) -- ($(\x2,\y3) + (dwn)$) -- (\x2,\y3);
\draw let \p1=(n3.west), \p2=(n5.east), \p3=(bot) in (\x1,\y3) -- ($(\x1,\y3) + (dwn)$) -- ($(\x2,\y3) + (dwn)$) -- (\x2,\y3);
\draw let \p1=(n1.west), \p2=(n5.east), \p3=($(bot) + (dwn)$) in (\x1,\y3) -- ($(\x1,\y3) + (dwn)$) -- ($(\x2,\y3) + (dwn)$) -- (\x2,\y3);
\end{tikzpicture}
\end{minipage}
\caption{Left-associative syntax tree (left) vs equal-level one (right) of the plus operator. The left syntax tree imposes a sequential left-to-right parsing and semantic processing whereas the right one can be split onto several branches to be partially processed independently and further aggregated.}
\end{figure}\label {fig:ass vs flat plus}

The answer to this question comes exactly from the above mentioned hypothesis: it forbids cyclic sequences of symbols that are at the same level of precedence; so that $+$ 
cannot be at the same level of itself. What was a, deemed irrelevant, theoretical limit hid a conceptual inadequacy which had not practical impact until OPLs properties have been exploited to build parallel compilers 
\footnote{It must also be recalled that in programming languages very long arithmetic expressions are quite unusual; not so in data description languages such JSON.}
We felt therefore ``compelled'' to finally remove this relatively disturbing hypothesis: this is the object of the present paper.

As often claimed in previous literature, when discussing the impact of forbidding cycles of operators all at the same level of precedence, removing such a restriction requires allowing grammar \emph{right hand sides (rhs) to include the Kleene $^*$ operator} (or the $^+$ one which excludes the empty string from its result).
Thus, we introduce the \emph{Cyclic operator-precedence grammars (C-OPG)} which include the above feature: 
whereas such a feature is often used in general context-free grammars to make the syntax of programming languages more compact but does not increase their expressive power, 
we show that C-OPGs are now fully equivalent to OPAs and other formalisms to define OPLs, such as the MSO logic.
We also show that all results previously obtained under the above hypothesis still hold by using C-OPGs instead of traditional OPGs.
Although the goal of this paper is not to develop parallel compilation algorithms rooted in C-OPGs, we show how they naturally overtake the difficulty pointed out by \cite{DBLP:conf/hpcasia/LiT23} and would allow to revisit their techniques, or to improve the efficiency of our previous parallel parser \cite{BarenghiEtAl2015}.

\section{Background}
We assume some familiarity with the classical literature on formal language and automata theory, e.g., \cite{Salomaa73,Harrison78}.
Here, we  just list and explain our notations for the basic concepts we use from this theory. 
The terminal alphabet is usually denoted by $\Sigma$, and the empty string is $\varepsilon$.
For a string, or set, $x$, $|x|$ denotes the length, or the cardinality, of $x$.
The character $\#$, not present in the terminal alphabet, is used as string \textit{delimiter}, and we define the alphabet  $\Sigma_\# = \Sigma \cup \{\# \}$.

\subsection{Regular languages: automata, regular expressions, logic}\label{subsec:FA}
A \emph{finite automaton} (FA) $\mathcal{A}$ is defined  by a  5-tuple
 $( Q,\Sigma, \delta, I, F)$  where
$Q$ is the set of states, $\delta$ the \emph{state-transition relation} (or its \emph{graph} denoted by $\longrightarrow$),  $\delta\subseteq Q \times \Sigma \times Q$;
$I$ and $F$ are the nonempty subsets of $Q$ respectively comprising the  initial  and  final states. 
If the tuple $(q,a,q')$ is  in the relation $\delta$, the edge $q\xlongrightarrow {a} q'$ is in the graph. 
The transitive closure of the relation is defined as usual. 
Thus, for a string $x \in \Sigma^*$ such that there is a path from state $q$ to $q'$ labeled with $x$, the notation $q\xlongrightarrow  {x} q'$ is equivalent to    
 $(q,x,q') \in \delta^*$;
if $q \in I$ and $q'\in F$, then the string  $x$ is {\em accepted}\/ by $\mathcal{A}$, and $L(\mathcal{A})$
is the language of the accepted strings.

A \emph{regular expression} (RE) over an alphabet $\Sigma$ is a well-formed formula made with the characters of $\Sigma$, $\emptyset$, $\varepsilon$, the Boolean  operators $\cup, \neg, \cap$, the concatenation  $\cdot$, and the Kleene star operator $^*$. 
We also use the operator $^+$, and we assume that its scope is always marked by parentheses.
An RE $E$ defines a language over $\Sigma$, denoted by $L(E)$, in the natural way.

Finite automata and regular expressions define the same family of languages named \emph{regular} (or rational)  languages (REG).

\subsection{Grammars}

\begin{definition}[Grammar and language]\label{def:grammar}
A \emph{context-free (CF) grammar} is a tuple $G=(\Sigma, V_N,  P, S)$ where $\Sigma$ and $V_N$, with $ \Sigma \cap V_N = \emptyset$, are resp.\ the terminal and the nonterminal alphabets, the total alphabet is $V = \Sigma \cup V_N$, $P\subseteq V_N \times V^*$ is the rule (or production) set,  and $S \subseteq V_N$, $S \neq \emptyset$,  is the axiom set.
For a generic rule, denoted as $A \to \alpha$,  where $A$ 
and $\alpha$ are resp.\ called the left/right hand sides (lhs / rhs), the following forms  are relevant:
\begin{center}
\begin{tabular}{l@{\;:\; }p{10cm}}
axiomatic & $A \in S$
\\
terminal & $\alpha \in \Sigma^+$
\\
empty & $\alpha =\varepsilon$
\\
renaming & $\alpha  \in V_N$
\\ 
operator & $\alpha \not\in V^* V_N V_N V^*$, i.e., at least one terminal is interposed between any two nonterminals occurring in $\alpha$
\\
\end{tabular}
\end{center}

\par
A grammar is called \emph{backward deterministic} or a BD-grammar  (or \emph{invertible}) if $(B \to \alpha, C \to \alpha \in P)$ implies $B=C$. 
\par
If all rules of a grammar  are in operator form, the grammar is called an \emph{operator grammar} or O-grammar.

For brevity we give for granted the usual definition of \emph{derivation} denoted by the symbols $\xLongrightarrow[G]{}$ (immediate derivation), $\xLongrightarrow[G]{\ast}$ (reflexive and transitive closure of $\xLongrightarrow[G]{}$),
$\xLongrightarrow[G]{+}$ (transitive closure of $\xLongrightarrow[G]{}$),
$\xLongrightarrow[G]{m}$ (derivation in $m$ steps);
 the subscript $G$ will be omitted whenever clear from the context. 
 We give also for granted the notion of \emph{syntax tree (ST)}. 
 As usual, the {\em frontier} of a syntax tree is the ordered left to right sequence of the leaves of the tree.
\par
The \emph{language} defined by a grammar starting from a nonterminal  $A$ is 
$L_G (A) = \left\{w  \mid w \in \Sigma^*, A \xLongrightarrow[G]{\ast} w \right\}.$

We call $w$ a \emph{sentence} if $A \in S$.  
The union of  $L_G(A)$ for all $A\in S$ is the language $L(G)$ defined by $G$. 

Two grammars defining the same language are \emph{equivalent}. Two grammars generating the same set of syntax trees are \emph{structurally equivalent}.
\end{definition}

\noindent {\em Notation:} In the following, \emph{unless otherwise explicitly stated}, lowercase letters at the beginning of the alphabet will denote terminal symbols, lowercase letters at the end of the alphabet will denote strings of terminals, 
Greek letters at the beginning of the alphabet will denote strings in $V^*$. Capital letters will be used for nonterminal symbols.

Any grammar can be effectively transformed into an equivalent BD-grammar,
and also into an O-grammar~\cite{DBLP:reference/hfl/AutebertBB97,Harrison78} without renaming rules and without empty rules but possibly a single rule whose lhs is an axiom not otherwise occurring in any other production.
\emph{From now on, w.l.o.g., we exclusively deal with O-grammars without renaming and empty rules, with the only exception that, if $\varepsilon$ is part of the language, there is a unique empty rule whose lhs is an axiom that does not appear in the rhs of any production.}

\subsection{Operator precedence grammars}
We define the  operator precedence grammars (OPGs) following primarily \cite{DBLP:journals/csr/MandrioliP18}. 

Intuitively, operator precedence grammars are O-grammars whose parsing is driven by three \emph{precedence relations}, called \emph{equal}, \emph{yield} and \emph{take}, included in $\Sigma_\# \times \Sigma_\#$.
They are defined in such a way that two consecutive terminals of a grammar's rhs
---ignoring possible nonterminals in between---
are in the equal relation, while the two extreme ones ---again, whether or not preceded or followed by a nonterminal--- 
are preceded by a yield and followed by a take relation, respectively; 
in this way a complete rhs of a grammar rule is identified and can be \emph{reduced} to a corresponding lhs by a typical bottom-up parsing.
More precisely, the three relations are defined as follows. Subsequently we show how they can drive the bottom-up parsing of sentences.

\begin{definition}  
[\cite{Floyd1963}]\label{def:OP-gramm}\label{def:OPG}
Let $G=(\Sigma, V_N, P, S)$ be an O-grammar.
Let $a,b$ denote elements in $\Sigma$,  $A, B$ in $V_N$, $C$ either an element of $V_N$ or the empty string $\varepsilon$, and $\alpha, \beta$ range over  $V^*$.
The \textit{left and right terminal sets} of nonterminals are respectively:
\[
  \mathcal{L}_G(A)  = \left\{a \in \Sigma \mid \exists C: A \xLongrightarrow[G]{\ast} C a \alpha
  \right\}
  \text{ and }
  \ \mathcal{R}_G(A)  =  \left\{a \in \Sigma \mid \exists C: A \xLongrightarrow[G]{\ast} \alpha a C \right\}.
\]
(The grammar name will be omitted unless necessary to prevent confusion.) 

The \emph{operator precedence (OP) relations} are defined over $\Sigma_\# \times \Sigma_\# $ 
as follows:
\begin{itemize}
\item equal in precedence:
  $
    a\doteq b  \iff 
    \exists A\to\alpha aCb\beta \in P \ \ 
    $
\item takes precedence:
$  a\gtrdot b  \iff \exists 
  A\to\alpha B b\beta \in P,
  a\in \mathcal{R}(B); 
  \\
  a\gtrdot \#  \iff  a\in \mathcal{R}(B), B \in S$
\item yields precedence:
  $
  a\lessdot b \iff  \exists  
  A\to\alpha aB\beta \in P, 
	b\in \mathcal{L}(B); \\
  \# \lessdot b \iff b\in \mathcal{L}(B), B \in S.
$
\end{itemize}

The OP relations can be collected  into a  $|\Sigma_\#| \times  |\Sigma_\#|$ array, called the \emph{operator precedence matrix} of the grammar,  $OPM(G)$: 
for each (ordered) pair $(a,b)\in \Sigma_\# \times  \Sigma_\#$,  $OPM_{a,b}(G)$ contains the OP relations holding between $a$ and $b$. 
\end{definition}

Consider a square matrix:
$M=\left\{M_{a,b}\subseteq \left\{\doteq, \lessdot, \gtrdot \right\} \,\mid \, a, b \in \Sigma_\# \right\}$.
 Such a matrix  is  called  \emph{conflict-free} iff $\forall a,b \in \Sigma_\#$, $0 \leq |M_{a,b}|\leq 1$. A conflict-free matrix  is called \emph{total} or \emph{complete} iff $\forall a,b \in \Sigma_\#$, $|M_{a,b}|=1$. By convention, if $M_{\#,\#}$ is not empty, $M_{\#,\#} = \{\doteq\}$.
A matrix  is $\dot=$-\emph{acyclic} if the transitive closure of the $\dot=$ relation over $\Sigma \times \Sigma$ is irreflexive.

We extend the set inclusion relations and the Boolean operations in the obvious cell by cell way, to any two matrices having the same terminal alphabet.
Two matrices are \emph{compatible} iff their union is conflict-free.

\begin{definition}[Operator precedence grammar]\label{def:OPgrammar}
A grammar $G$ is an \emph{operator precedence grammar} (OPG) iff the matrix  $OPM(G)$ is conflict-free, i.e.\ the three OP relations are disjoint.
An OPG is $\dot=$-\emph{acyclic} if $OPM(G)$ is so. 
An \emph{operator precedence language (OPL)} is a language generated by an OPG.
\end{definition}

Figure \ref{fig:OPM1} (left) displays an  OPG, $G_{AE}$, which generates simple, unparenthesized arithmetic expressions and its OPM (center). 
The left and right terminal sets
of $G_{AE}$'s nonterminals $E$, $T$ and $F$ are, respectively:
$\mathcal{L}(E)  = \{+, \times,n\}$,
$\mathcal{L}(T)  = \{\times,n\}$, 
$\mathcal{L}(F)  = \{n\}$,
$\mathcal{R}(E) =  \{+, \times,n\}$,
$\mathcal{R}(T) =  \{\times,n\}$, and 
$\mathcal{R}(F) =  \{n\}$. 
  
\noindent\emph{Remark.} If the relation $\dot=$ is acyclic, then the length of the rhs of any rule of $G$ is bounded by the length of the longest $\dot=$-chain in $OPM(G)$. 

\begin{figure}
\begin{tabular}{m{0.4\textwidth}m{0.3\textwidth}m{0.2\textwidth}}
$
 	\begin{array}{ll}
	G_{AE}: &	S = \{E\} \\
 	&	E \to  E + T \mid T \times F \mid n  \\  	
 	&	T \to  T \times F \mid n      \\
 	&	F \to  n  \\
 	\end{array}
$ &
$
\begin{array}{c|cccc}
    &+ & \times  & n  & \#\\
\hline
+ & \gtrdot & \lessdot & \lessdot & \gtrdot\\
\times & \gtrdot & \gtrdot  & \lessdot & \gtrdot\\
n & \gtrdot & \gtrdot  &         & \gtrdot \\
\# & \lessdot & \lessdot &\lessdot \\
\end{array}
$ &
\centering
\begin{tikzpicture}[scale=0.4]
\node{$N$}
child{ node{$\#$} }
child{
node{$N$}
child{ node{$N$} 
    child{ node{$N$} child { node{$n$} } }
    child{ node{$+$} }
    child{ node{$N$} 
        child{ node{$N$} child { node{$n$} }}
        child{ node{$\times$} }
        child{ node{$N$} child { node{$n$} }}
    }
}
child{ node{$+$} }
child{ node{$N$} child { node{$n$} } }
}
child{ node{$\#$} };
\end{tikzpicture}
\end{tabular}
\caption{$G_{AE}$ (left), its OPM (center), and the syntax tree of $n + n \times n + n$ according to the OPM (right).}
\label{fig:OPM1}
\label{fig:ST}
\end{figure}
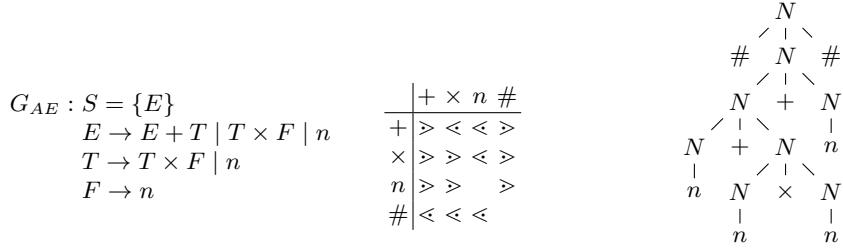

Unlike the  arithmetic relations having similar typography, the OP relations do not enjoy any of the transitive, symmetric, reflexive properties. We kept the original Floyd's notation but we urge the reader not to be confused by the similarity of the two notations.

It is known that the family of OPLs is strictly included within the deterministic and reverse-deterministic CF family, i.e., the languages that can be deterministically parsed both from left to right and from right to left. 

The key feature of OPLs is that a conflict-free OPM $M$ defines a universe of \emph{strings compatible with $M$} and associates to each of them a unique \emph{syntax tree} whose internal nodes are unlabeled and whose leaves are elements of $\Sigma$, or, equivalently, a unique parenthesization. We illustrate such a feature through a simple example and refer the reader to previous literature  for a thorough description of OP parsing \cite{GruneJacobs:08,DBLP:journals/csr/MandrioliP18}.

\begin{example}\label{ex:OPG:GAE_1}
Consider the $OPM(G_{AE})$ of Figure \ref{fig:OPM1} and the string $n + n \times n + n$. Display all precedence relations holding between consecutive terminal characters, \emph{including the relations with the delimiters \#} as shown below:
\[ 
\# \lessdot n \gtrdot + \lessdot n \gtrdot \times \lessdot  n \gtrdot  + \lessdot n \gtrdot  \#
\]
each pair $ \lessdot , \gtrdot $ (with no further $ \lessdot , \gtrdot $ in between) includes a \emph{possible} rhs of a production of \emph{any OPG} sharing the OPM with $G_{AE}$, not necessarily a $G_{AE}$ rhs. 
Thus, as it happens in typical bottom-up parsing, we replace
---\emph{possibly in parallel}---
each string included within the pair $ \lessdot , \gtrdot $ with a \emph{dummy nonterminal $N$}; this is because nonterminals are irrelevant for OPMs. 
The result is the string $\# N + N \times N + N \# $. Next, we compute again the precedence relation between consecutive terminal characters by \emph{ignoring nonterminals}: the result is 
$ \# \lessdot N + \lessdot N \times N \gtrdot  +  N \gtrdot  \# $.

This time, there is only one pair  $ \lessdot , \gtrdot $ including a potential rhs determined by the OPM (the fact that the external $ \lessdot $ and  $ \gtrdot $ ``look matched" is coincidental as it can be easily verified by repeating the previous procedure with the string $n + n \times n + n + n$).
Again, we replace the pattern $N \times N$, with the dummy nonterminal $N$; notice that there is no doubt about associating the two $N$ to the $\times$ rather than to one of the adjacent $+$ symbols: 
if we replaced, say, just the $\times$ with an $N$ we would obtain the string
$N + NNN + N$ which cannot be derived by an O-grammar.
By recomputing the precedence relations we obtain the string $\# \lessdot N + N \gtrdot + N \gtrdot \# $.
Finally, by applying twice the replacing of $N + N$ by $N$ we obtain $ \# N\# $. 

The result of the whole bottom-up reduction procedure is synthetically represented by the \emph{syntax tree} of Figure \ref{fig:ST} (right) which shows the precedence of the multiplication operation over the additive one in traditional arithmetics.
It also suggests a natural association to the left of both operations: 
if we reverted the order of the rhs of the rules rewriting $E$ and $T$, the structure of the tree would have suggested associativity to the right of both operations which would not have altered the semantics of the two operations which can indifferently be associated to the left and to the right; 
not so is we dealt with, say, subtraction or division which instead \emph{impose} association to the left.

Notice that the tree of Figure \ref{fig:ST} has been obtained ---uniquely and determin\-istically--- by using exclusively the OPM, not the grammar $G_{AE}$ although the string $n + n \times n + n \in L(G_{AE})$
\footnote{The above procedure that led to the syntax tree of Figure \ref{fig:ST} could be easily adapted to become an algorithm that produces a new syntax tree whose internal nodes are labeled by $G_{AE}$'s nonterminals.
Such an algorithm could be made deterministic by transforming $G_{AE}$ into an equivalend BD grammar (sharing the same OPM).}.

Obviously, all sentences of $L(G_{AE})$ can be given a syntax tree by $OPM(G_{AE})$, but there are also strings in $\Sigma^*$ that can be parsed according to the same OPM but are not in $L(G_{AE})$. 
E.g., the string $+++$ is parsed according to the $OPM(G_{AE})$ as a ST that associates the $+$ characters to the left. 
Notice also that, in general, not every string in $\Sigma^*$ is assigned a syntax tree by an OPM; e.g., in the case of $OPM(G_{AE})$ the parsing procedure applied to $nn$ is immediately blocked since there is no precedence relation between $n$ and itself.
\end{example}

The following definition synthesizes the concepts introduced by Example \ref{ex:OPG:GAE_1}.

\begin{definition}[OP-alphabet and Maxlanguage]\label{def:OP structures}
\begin{itemize}
\item
A string in $\Sigma^*$ is \emph{compatible} with an OPM $M$ iff the procedure described in Example~\ref{ex:OPG:GAE_1} terminates by producing the pattern $\#N\#$. The set of all strings compatible with an OPM $M$ is called the \emph{maxlanguage} or the \emph{universe} of $M$ and is simply denoted as $L(M)$.
\item
 Let $M$ be a conflict-free OPM over $\Sigma_\# \times \Sigma_\#$. We use the same identifier $M$ to denote the ---partial--- function $M$ that assigns to strings in $\Sigma^*$ their unique ST as informally illustrated in Example~\ref{ex:OPG:GAE_1}.
 \item 
 The pair $(\Sigma, M)$ where $M$ is a conflict-free OPM over $\Sigma_\# \times \Sigma_\#$, is called an \emph{OP-alphabet}. We introduce the concept of OP-alphabet as a pair to emphasize that it defines a universe of strings on the alphabet $\Sigma$ ---not necessarily covering the whole $\Sigma^*$--- and implicitly assigns them a structure univocally determined by the OPM, or, equivalently, by the function $M$.
 \item 
The class  of \emph{$(\Sigma, M)$-compatible} OPGs and OPLs are: 
\[
\mathscr{G}_M = \{G  \mid G \text{ is an OPG and } OPM(G) \subseteq M \},
\quad
 \mathscr{L}_M =  \{ L(G) \mid G \in \mathscr{G}_M\}. 
\]
\end{itemize}
\end{definition}

Various formal properties of OPGs and OPLs are documented in the literature, e.g., in~\cite{Crespi-ReghizziMM1978,Crespi-ReghizziM12,DBLP:journals/csr/MandrioliP18}. The next proposition recalls those that are relevant for this article.

\begin{proposition}[Algebraic properties of OPGs and OPLs]\label{propo:OPGandOPL}\label{PropOPL}
\begin{enumerate}
\item 
If an OPM $M$ is total, then the corresponding homonymous function is total as well, i.e., $L(M)= \Sigma^*$.

\item
Let $(\Sigma, M)$ be an OP-alphabet where $M$ is $\dot=$-acyclic. The class $\mathscr{G}_M$ contains an OPG, called the \emph{maxgrammar} of $M$, denoted by $G_{max, M}$, which generates the maxlanguage $L(M)$. 
For all grammars $G \in \mathscr{G}_M$, $L(G) \subseteq L(M)$.

\item The closure properties of the family $\mathscr{L}_M$ of $(\Sigma, M)$-compatible OPLs defined by a total OPM are the following: 
\begin{itemize}
\item $\mathscr{L}_M$ is closed under union, intersection and set-difference, therefore also under complement (if a maxgrammar of $M$ exists).
\item $\mathscr{L}_M$ is closed under concatenation.
\item if $M$ is $\dot=$-acyclic, $\mathscr{L}_M$ is closed under Kleene star.
\end{itemize}
\end{enumerate}
\end{proposition}

\noindent \emph{Remark}. Thanks to the fact that a conflict-free OPM assigns to each string at most one ST 
---and exactly one if the OPM is complete--- 
the above closure properties of OPLs w.r.t.\ Boolean operations automatically extend to sets of their STs.
The same does not apply to the case of concatenation which in general may produce significant reshaping of the original STs \cite{Crespi-ReghizziM12}.
Furthermore, any complete, conflict-free, $\dot=$-acyclic OPM defines a \emph{universe of STs} whose frontiers 
are $\Sigma^*$. 

\subsection{Chains and Operator Precedence Automata (OPA)}
The notion of \emph{chain} introduced next is an alternative way to represent STs where internal nodes are irrelevant and ``anonymized''.
\begin{definition}[Chains]
	Let $(\Sigma, M)$ be an OP-alphabet.
	\begin{itemize}
		\item A \emph{simple chain} is a word $a_0 a_1 a_2 \dots$ $a_n a_{n+1}$,
		written as
		$
		\chain {a_0} {a_1 a_2 \dots a_n} {a_{n+1}},
		$
		such that:
		$a_0, a_{n+1} \in \Sigma \cup \{\#\}$, 
		$a_i \in \Sigma$ for every $i: 1\leq i \leq n$, 
		$M_{a_0 a_{n+1}} \neq \emptyset$,
		and $a_0 \lessdot a_1 \doteq a_2 \dots a_{n-1} \doteq a_n \gtrdot a_{n+1}$.
		
		\item A \emph{composed chain} is a word 
		$a_0 x_0 a_1 x_1 a_2  \dots a_n x_n a_{n+1}$, with $x_i \in \Sigma^*$, 
		where\linebreak
		$\chain {a_0}{a_1 a_2 \dots a_n}{a_{n+1}}$ is a simple chain, and
		either $x_i = \varepsilon$ or $\chain {a_i} {x_i} {a_{i+1}}$ is a chain (simple or composed),
		for every $i: 0\leq i \leq n$. 
		Such a composed chain will be written as
		$\chain {a_0} {x_0 a_1 x_1 a_2 \dots a_n x_n} {a_{n+1}}$.
		\item The \textit{body} of a chain $\chain axb$, simple or composed, is the word $x$.
	\item Given a chain  $\chain {a} {x} {b}$  the \emph{depth} $d(x)$ of its body $x$ is defined recursively: 
	$d(x) = 1$ if the chain is simple, whereas 
	$d(x_0 a_1 x_1 \dots a_n x_n) = 1 + \max_i d(x_i)$. The depth of a chain is the depth of its body.
	\end{itemize}
\end{definition}

For instance, the ST of Figure \ref{fig:ST} (right) is biunivocally represented by the composed chain $\chain {\#} {x_0 + x_1} {\#}$,
where, in turn $x_0$ is the body of the composed chain $\chain {\#} {y_0 + y_1} {+}$, $y_0$ is the body of the simple chain 
$\chain {\#} {e} {+}$, $y_1$ is the body of the composed chain $\chain {+} {z_0 \times z_1} {+}$, etc. The depth of the main chain is $5$.

As well as an OPG selects a set of STs within the universe defined by its OPM, an \emph{operator precedence automaton (OPA)} selects a set of chains within the universe defined by an OP-alphabet.

\begin{definition}[Operator precedence automaton]\label{def:OPA}
	A nondeterministic \emph{operator precedence automaton (OPA)} is given by a tuple:
	$\mathcal A = \langle \Sigma, M, Q, I, F, \delta \rangle $ where:
	\begin{itemize}
		\item $(\Sigma, M)$ is an operator precedence alphabet,
		\item $Q$ is a set of states (disjoint from $\Sigma$),
		\item $I \subseteq Q$ is a set of initial states,
		\item $F \subseteq Q$ is a set of final states,
		\item $\delta$, named transition function, is a triple of functions:
		\[
	\delta_{\text{shift}}: Q \times \Sigma \rightarrow \powerset(Q)
		\qquad 
		\delta_{\text{push}}: Q \times \Sigma \rightarrow \powerset(Q)
		\qquad 
		\delta_{\text{pop}}: Q \times Q \rightarrow \powerset(Q)
		\]
	\end{itemize}
	
\end{definition}

We represent a nondeterministic OPA by a graph with $Q$ as the set of vertices and
$\Sigma \cup Q$ as the set of edge labelings. 
The edges of the graph are denoted by different shapes of arrows to distinguish the three types of transitions:
there is an edge from state $q$ to state $p$ labeled by $a \in \Sigma$ denoted by a dashed (respectively, normal) arrow if and only if $p \in \delta_{\text{shift}}(q,a)$ (respectively, $p \in \delta_{\text{push}}(q,a))$ and 
there is an edge from state $q$ to state $p$ labeled by $r \in Q$ and denoted by a double arrow if and only if $p \in \delta_{\text{pop}}(q,r)$.

To define the semantics of the automaton, we introduce some notations.

We use letters $p, q, p_i, q_i, \dots $ to denote states in $Q$.
Let $\Gamma$ be	$\Sigma \times Q$ and let $\Gamma'$ be $\Gamma \cup  \{\bot\} $; 
we denote symbols in $\Gamma'$ as $\tstack aq$ or $\bot$.
We set $\symb {\tstack aq} = a$, $\symb {\bot}=\#$, and
$\state {\tstack aq} = q$.
Given a string $\Pi = \bot \pi_1 \pi_2 \dots \pi_n $, with $ \pi_i \in \Gamma$ , $n \geq 0$, 
we set $\symb \Pi = \symb{\pi_n}$, including the particular case  $\symb \bot = \#$.

A \emph{configuration} of an \opa\ is a triple $C = \tconfig \Pi q w$,
where $\Pi \in \bot\Gamma^*$, $q \in Q$ and $w \in \Sigma^*\#$.
The first component represents the contents of the stack, the second component represents the current state of the automaton, while the third component is the part of input still to be read.

A \emph{computation} or \emph{run} of the automaton is a finite sequence of \emph{moves} or \emph{transitions} 
$C_1 \vdash C_2$; 
there are three kinds of moves, depending on the precedence relation between the symbol on top of the stack and the next symbol to read:

\smallskip
\noindent {\bf push move:} if $\symb \Pi \lessdot \ a$ then
$
\tconfig \Pi p {ax} \vdash \tconfig {\Pi \tstack a p} q x$, with $q \in \delta_{\text{push}}(p,a)$;

\smallskip
\noindent {\bf shift move:} if $a \doteq b$ then 
$
\tconfig {\Pi \tstack a p} q {bx} \vdash \tconfig {\Pi \tstack b p} r x$, with $r \in \delta_{\text{shift}}(q,b)$;

\smallskip

\noindent {\bf pop move:} if $a \gtrdot b$
then $
\tconfig {\Pi \tstack a p} q {bx} \vdash \tconfig \Pi r {bx}$, with $r \in \delta_{\text{pop}}(q, p)$.

Notice that shift and pop moves are never performed when the stack contains only $\bot$.

Push and shift moves update the current state of the automaton according to the transition function $\delta_{\text{push}}$ and  $\delta_{\text{shift}}$, respectively: push moves put a new element on the top of the stack consisting of the input symbol together with the current state of the automaton, whereas shift moves update the top element of the stack by changing its input symbol only.
The pop move removes the symbol on the top of the stack,
and the state of the automaton is updated by $\delta_{\text{pop}}$ on the basis of the pair of states consisting of the current state of the automaton and the state of the removed stack symbol; notice that in this move the input symbol is used only to establish the $\gtrdot$ relation and it remains available for the following move.

A configuration $\tconfig {\bot}{q_I}{x\#}$ is {\em initial} if $q_I \in I$;
a  configuration $\tconfig {\bot}{q_F}{\#}$ is {\em accepting} if $q_F \in F$.
The language accepted by the automaton is:
\[
L(\mathcal A) = \left\{ x \mid  \tconfig{\bot} {q_I}{x\#}  \comp * 
\tconfig {\bot}{q_F}{\#} , q_I \in I, q_F \in F \right\}.
\]

\begin{example}\label{ex:exprAut}
 The OPA depicted in Figure~\ref{fig:exprAut} (top, left) based on the OPM at the (top, right) accepts the language of arithmetic expressions enriched w.r.t  $L(G_{AE})$ in that it introduces the use of explicit parentheses to alter the natural precedence of arithmetic operations.
	The same figure (bottom) also shows an accepting computation on input $n + n \times \lp n + n \rp$.
\end{example}

\begin{figure}[tb]
\centering
\begin{tabular}{m{0.4\textwidth}m{0.4\textwidth}}
    \begin{tikzpicture}[every edge/.style={draw,solid}, node distance=4cm, auto, 
	every state/.style={draw=black!100,scale=0.5}, >=stealth]
		
    \node[initial by arrow, initial text=,state] (q0) {{\huge $q_0$}};
    \node[state] (q1) [right of=q0, xshift=0cm, accepting] {{\huge $q_1$}};
    \node[state] (q2) [below of=q0, xshift=0cm] {{\huge $q_2$}};
    \node[state] (q3) [right of=q2, xshift=0cm, accepting] {{\huge $q_3$}};
    
    \path[->]
    (q0) edge [below, \parrow] node {$n$} (q1)
    (q0) edge [bend right, right, \parrow]  node {$\lp$} (q2)
    (q1) edge [loop above, double, right]  node {$\ q_0, q_1$} (q1)
    (q1) edge [bend right, above, \parrow]  node {$+, \times$} (q0)
    
    (q2) edge [below,  \parrow] node {$n$} (q3)
    (q2) edge [loop left, left, \parrow] node {$\lp\ $} (q2)
    (q3) edge [loop above, right, double]  node {$\ q_0, q_1, q_2, q_3$} (q3) 
    (q3) edge [bend right, above, \parrow]  node {$+, \times$} (q2)
    (q3) edge [loop right, right, \iarrow] node {$\ \rp$} (q3);
	\end{tikzpicture}
	&
    \[
    \begin{array}{c|cccccc}
    &+ &\times & \lp & \rp & n  & \# \\
    \hline
    + & \gtrdot &\lessdot &\lessdot &\gtrdot &\lessdot &\gtrdot \\
    \times  & \gtrdot &\gtrdot &\lessdot &\gtrdot &\lessdot &\gtrdot\\
    \lp & \lessdot &\lessdot &\lessdot &\doteq &\lessdot \\
    \rp & \gtrdot &\gtrdot & &\gtrdot & &\gtrdot\\
    n & \gtrdot &\gtrdot & &\gtrdot & &\gtrdot\\
    \# & \lessdot & \lessdot& \lessdot && \lessdot
    \end{array}
    \]
\end{tabular}
\begin{scriptsize}
    \[
    \begin{array}{|l|c|r|}
    \hline
    \text{stack}   											 
    &  \text{state}
    &  \text{current input}\\\hline
    \bot &  q_0  &  n + n \times \lp n + n \rp \# \\ \hline 
    \bot  \tstack {n} {q_0}      &   q_1  &  + n \times \lp n + n \rp  \#  \\ \hline
    \bot &  q_1  &  + n \times \lp n + n \rp  \# \\ \hline 
    \bot  \tstack {+} {q_1}    &    q_0  &  n \times \lp n + n \rp  \#  \\ \hline
    \bot  \tstack {+} {q_1}\tstack  n {q_0}   &   q_1  &   \times \lp n + n \rp \#  \\ \hline
    \bot  \tstack  + {q_1}     &   q_1  &  \times \lp n + n \rp  \#  \\ \hline
    \bot  \tstack  + {q_1}    \tstack  \times {q_1}    &    q_0  &  \lp n + n \rp  \#  \\ \hline
    \bot  \tstack  + {q_1}    \tstack  \times {q_1}  \tstack  \lp {q_0}    &    q_2  &  n + n \rp  \#  \\ \hline
    \bot  \tstack  + {q_1}    \tstack  \times {q_1}  \tstack  \lp {q_0}  \tstack  n {q_2}   &  q_3  &   + n \rp  \#  \\ \hline
    \bot  \tstack  + {q_1}    \tstack  \times {q_1}  \tstack  \lp {q_0}    &  q_3  &    + n \rp  \#  \\ \hline
    \bot  \tstack  + {q_1}    \tstack  \times {q_1}  \tstack  \lp {q_0}  \tstack  + {q_3}  &   q_2  &  n \rp  \#  \\ \hline
    \bot  \tstack  + {q_1}    \tstack  \times {q_1}  \tstack  \lp {q_0}  \tstack  + {q_3} \tstack  n {q_2}  &   q_3  &  \rp  \#  \\ \hline
    \bot  \tstack  + {q_1}    \tstack  \times {q_1}  \tstack  \lp {q_0}  \tstack  + {q_3}  &  q_3  &   \rp  \#  \\ \hline
    \bot \tstack  + {q_1}    \tstack  \times {q_1}  \tstack  \lp {q_0}    &   q_3  &  \rp  \#  \\ \hline
    \bot \tstack  + {q_1}    \tstack  \times {q_1}  \tstack  \rp {q_0}    &  q_3  &     \#  \\ \hline
    \bot  \tstack  + {q_1}    \tstack  \times {q_1}   &   q_3  &    \#  \\ \hline
    \bot  \tstack  + {q_1}    &  q_3  &    \#  \\ \hline
    \bot &   q_3  &  \#  \\ \hline
    \end{array}
    \]
\end{scriptsize}
\caption{An OPA (top, left), its OPM (top, right) and an example of computation for the language of Example~\ref{ex:exprAut} (bottom). Arrows $\va{}{}$, $\shift{}$ and $\flush{}{}$ denote push, shift and pop transitions, respectively. To avoid confusion with the overloaded parenthesis symbols, the parentheses used as terminal symbols are denoted as $\lp$ and $\rp$.}
\label{fig:exprAut}
\end{figure}
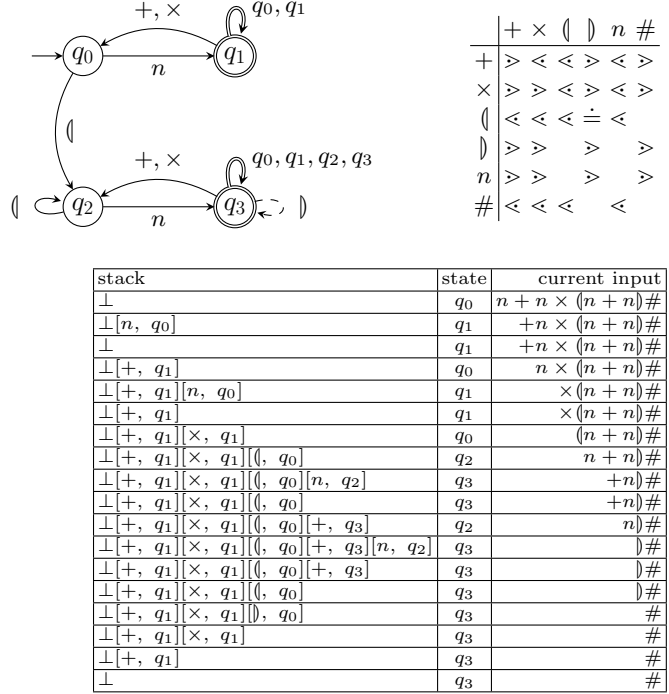 

\medskip

\begin{definition} 
	Let $\mathcal A$ be an \opa .
	A \emph{support} for a simple chain
	$\chain {a_0} {a_1 a_2 \dots a_n} {a_{n+1}}$
	is any path in $\mathcal A$ of the form
	\begin{equation}
	\label{eq:simplechain}
	q_0
	\va{a_1}{q_1}
	\shift{}{}
	\dots
	\shift{}q_{n-1}
	\shift{a_{n}}{q_n}
	\flush{q_0} {q_{n+1}}
	\end{equation}
	Notice that the label of the last (and only) pop is exactly $q_0$, i.e.\ the first state of the path; this pop is executed because of relations $a_0 \lessdot a_1$ and $a_n
	\gtrdot a_{n+1}$.
	
	\noindent A \emph{support for the composed chain} 
	$\chain {a_0} {x_0 a_1 x_1 a_2 \dots a_n x_n} {a_{n+1}}$
	is any path in $\mathcal A$ of the form
	\begin{equation}
	\label{eq:compchain}
	q_0
	\ourpath{x_0}{q'_0}
	\va{a_1}{q_1}
	\ourpath{x_1}{q'_1}
	\shift{a_2}{}
	\dots
	\shift{a_n} {q_n}
	\ourpath{x_n}{q'_n}
	\flush{q'_0}{q_{n+1}}
	\end{equation}
	where for every $i: 0\leq i \leq n$: 
	\begin{itemize}
		\item if $x_i \neq \varepsilon$, then ${q_i} \ourpath{x_i}{q'_i} $ 
		is a support for the (simple or composed) chain $\chain {a_i} {x_i} {a_{i+1}}$
		
		\item if $x_i = \varepsilon$, then $q'_i = q_i$.
	\end{itemize}
	\noindent Notice that the label of the last pop is exactly $q'_0$.
	\\
	The support of a chain with body $x$ will be denoted by ${q_0} \ourpath{x}{q_{n+1}}$.
\end{definition}

Notice that the context $ a, b$ of a chain $\chain axb$ is used by the automaton to build its support only because $ a \lessdot x$ and $x \gtrdot b$; thus, the chain's body contains all
information needed by the automaton to build the subtree whose frontier is that string, once it is understood that its first move is a push and its last one is pop.
This is a distinguishing feature of \opl s, not shared by other
deterministic languages: 
we call it the \emph{locality principle} of \opl s, which has been exploited to build parallel and/or incremental OP parsers \cite{BarenghiEtAl2012a}.

\section{Cyclic Operator Precedence Grammars (C-OPGs) and their equivalence with OPAs}\label{A-OPG vs OPA}

Proposition \ref{PropOPL} shows that \emph{some, but not all}, of the algebraic properties of OPLs depend critically on the $\dot=$-acyclicity hypothesis. 
This is due to the fact that without such a hypothesis the rhs of an OPG have an unbounded length but cannot be infinite: 
e.g., no OPG can generate the language $\{a,b\}^*$ if $a \doteq b$ and $b \doteq a$. In most cases cycles of this type can be ``broken'' 
as it has been done up to now, e.g., to avoid the $+ \doteq +$ relation in arithmetic expressions by associating the operator indifferently to the right or to left.
Thus, although it is known that, from a theoretical point of view, the $\dot=$-acyclicity hypothesis affects the expressive power of OPGs,%
\footnote{The language $\{a^n(bc)^n\} \cup \{b^n(ca)^n\} \cup \{c^n(ab)^n\} \cup (abc)^+$ cannot be generated by an OPG because the $a\doteq b \doteq c \doteq a$ relations are necessary \cite{HOP20}.}
we accepted so far the $\dot=$-acyclicity hypothesis to keep the notation as simple as possible. Recently, however, it has been observed \cite{DBLP:conf/hpcasia/LiT23} that such a restriction may hamper the benefits achievable by the parallel compilation techniques that exploit the local parsability property of OPLs \cite{BarenghiEtAl2015}.
Thus, it is time to introduce the necessary extension of OPGs so that the $\dot=$-acyclicity hypothesis can be avoided and they become fully equivalent to other formalisms to define OPLs.

\begin{definition}[Cyclic Operator Precedence Grammar]\label{def:A-OPG}
\begin{itemize}
\item 
A $^+$-O-expression on $V^*$ is an expression obtained from the elements of $V$ by iterative application of concatenation and the usual $^+$ operator
\footnote{For our purposes $^+$ is more convenient than $^*$ without affecting the generality.}, provided that any substring thereof has no two adjacent nonterminals; 
for convenience, 
and w.l.o.g., we assume that all subexpressions that are argument of the $^+$ operator are terminated by a terminal character.
    \item A Cyclic O-grammar (C-OG) is an O-grammar whose production rhs are $^+$-O-expressions.
    \item
    For a grammar rule $A \to \alpha$ of  an C-OG, the  $\xLongrightarrow[G]{}$ (immediate derivation) relation is defined as
    $\beta A \gamma \xLongrightarrow[]{} \beta \zeta \gamma$ iff $\zeta$ is a string belonging to the language defined by the $^+$-O-expression $\alpha$, $L(\alpha)$.
    \item 
    The equal in precedence relation  is redefined as $a\doteq b$  iff $\exists A \to \alpha 
     \land 
    \exists \zeta = \eta aBb \theta \mid  (B\in V_N\cup \{\varepsilon \}  
    \land \zeta \in L(\alpha))$.
    The other precedence relations remain defined as for non-cyclic O-grammars.
    \item 
    A C-OG is a cyclic operator precedence grammar (C-OPG) iff its OPM is conflict-free.
\end{itemize}
\end{definition}
As a consequence of the definition of the immediate derivation relation for C-OPGs the syntax-trees derived therefrom can be unranked, i.e., their internal nodes may have an unbounded number of children.
\begin{example}
\begin{figure}
\begin{minipage}{.7\linewidth}
\begin{align*}
&G_\text{A-AE}:\\
&
 	\begin{array}{ll}
	&	S = \{P, T, M, N, F, D, E\} \\
 	&	P\to  (T+)^+T  \\  	
 	&	T \to  (F\times)^+ F \mid M - N \mid D/E  \\
    & M \to M-N \mid (F\times)^+ F \mid D/E \\ 
    & N \to (F\times)^+ F \mid D/E \\ 
 	& F \to  D/E  \\
    & D \to  D/E  \\
    & \{P, T, M, N, F, D, E\} \to \lp \{P, T, M, N, F, D, E\} \rp \mid n
 	\end{array}
\end{align*}
\end{minipage}
\begin{minipage}{.29\linewidth}
\[
    \begin{array}{c|cccccccc}
    &+ & - &\times & / & \lp & \rp & n  & \# \\
    \hline
    + & \doteq &\lessdot &\lessdot &\lessdot &\lessdot &\gtrdot &\lessdot &\gtrdot\\
    - &\gtrdot &\gtrdot &\lessdot &\lessdot &\lessdot &\gtrdot &\lessdot &\gtrdot\\
    \times  & \gtrdot &\gtrdot & \doteq &\lessdot &\lessdot  &\gtrdot &\lessdot &\gtrdot\\
    / &\gtrdot &\gtrdot &\gtrdot &\gtrdot &\lessdot &\gtrdot &\lessdot &\gtrdot\\
    \lp & \lessdot &\lessdot &\lessdot & \lessdot &\lessdot &\doteq &\lessdot \\
    \rp & \gtrdot &\gtrdot & \gtrdot &\gtrdot & &\gtrdot & &\gtrdot\\
    n & \gtrdot &\gtrdot  &\gtrdot  &\gtrdot & &\gtrdot & &\gtrdot\\
    \# & \lessdot & \lessdot & \lessdot & \lessdot & \lessdot  && \lessdot
    \end{array}
\]
\end{minipage}
    \caption{A C-OPG (left) and its OPM (right). The notation $\{P,T,M,F,D,E\} \to \lp \{P, T, M, N, F, D, E\} \rp$ is an abbreviation meaning that anyone of the nonterminals at the left can be rewritten as a pair of parentheses enclosing anyone of the same nonterminals.}
    \label{fig:cyclic AE}
\end{figure}
The C-OPG depicted in Figure \ref{fig:cyclic AE} with its OPM generates a fairly complete language of ---not necessarily fully--- parenthesized arithmetic expressions involving the four basic operations: 
as usual the multiplicative operations take precedence over the additive ones; furthermore subtraction takes precedence over sum and division over multiplication.
The key novelty w.r.t. the traditional way of formalizing arithmetic expressions by means of OPGs are the $+ \doteq +$ and $\times \doteq \times$ OP relations; on the contrary we kept the structure that associates subtraction and division to the left, so that the grammar's STs ---an example thereof is given in Figure \ref{fig:cyclic ST}--- now fully reflect the semantics of the arithmetic operations.
\begin{figure}
\centering
\begin{tikzpicture}[
scale=.8,
level distance=3em,
level 1/.style={sibling distance=3em}
]
\node{$P$}
child{ node{$T$}
    [sibling distance=3em]
    child{ node{$M$}
        [sibling distance=1.5em]
        child{ node{$M$} child{ node{$n$} } }
        child{ node{$-$} }
        child{ node{$N$} child{ node{$n$} } }
    }
    child{ node{$-$} }
    child{ node{$N$}
        [sibling distance=1.5em]
        child{ node{$F$} child{ node{$n$} } }
        child{ node{$\times$} }
        child{ node{$F$} child{ node{$n$} } }
    }
}
child{ node{$+$} }
child{ node{$T$} child{ node{$n$ } } }
child{ node{$+$} }
child{ node[xshift=-.5em]{$T$}
    [sibling distance=1.5em]
    child{ node{$F$} child{ node{$n$} } }
    child{ node{$\times$} }
    child{ node{$F$}
        child{ node{$D$} child{ node{$n$} } }
        child{ node{$/$} }
        child{ node{$E$} child{ node{$n$} } }
    }
    child{ node{$\times$} }
    child{ node{$F$} child{ node{$n$} } }
}
child{ node{$+$} }
child{ node[xshift=.5em]{$T$}
    [sibling distance=1.5em]
    child{ node{$D$}
        child{ node{$\lp$} }
        child{ node{$P$}
            child{ node{$T$} child{ node{$n$} } }
            child{ node{$+$} }
            child{ node{$T$} child{ node{$n$} } }
        }
        child{ node{$\rp$} }
    }
    child{ node{$/$} }
    child{ node{$E$} child{ node{$n$} } }
}
child{ node{$+$} }
child{ node{$T$} child{ node{$n$} } }
child{ node{$+$} }
child{ node{$T$} child{ node{$n$} } };
\end{tikzpicture}
\caption{A ST generated by the C-OPG of Figure \ref{fig:cyclic AE}.}
\label{fig:cyclic ST}
\end{figure}
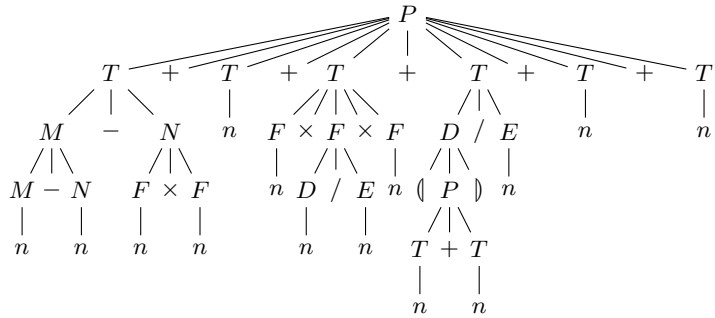
\end{example}\label{ex:cyclic ST}

Notice that $G_\text{A-AE}$ is purposely ambiguous to keep it compact. To support deterministic parsing it should be transformed into a BD ---and therefore unambiguous--- form.

By looking at the ST of Figure \ref{fig:cyclic ST} and comparing it with the original Figure \ref{fig:ass vs flat plus},
one  can get a better envision of why the introduction of cyclic $\doteq$ can support more effective parallel compilation algorithms for OPLs.
Parallel parsing and compilation for OPLs is rooted in the \emph{local parsability property} of this family: 
thanks to this property any fragment of input string enclosed within a pair of corresponding $\lessdot$ and $\gtrdot$ OP relations can be processed in parallel with other similar fragments.
However, if, say, the $+$ operator is associated to the left (or right) as in the case of the left ST of Figure~\ref{fig:ass vs flat plus} 
the parsing of a sequence of $+$ must necessarily proceed sequentially from left to right.
Conversely, if the ST has a structure like that of the right tree of Figure \ref{fig:ass vs flat plus} 
the sequence of $+$ ---whether intermixed or not with other subtrees--- can be arbitrarily split into several branches which can be parsed and compiled in parallel and, 
after that, can be joined into a unique subtree as imposed by the $\doteq$ OP relation between the corresponding extreme terminals of contiguous branches.

\subsection{Equivalence between C-OPGs and OPAs}
The equivalence is obtained by a slight modification of the analogous proof given in \cite{LonatiEtAl2015} where the additional hypothesis of $M$ being $\dot=$-acyclic was exploited.

First, we describe a procedure to build an OPA equivalent to a C-OPG. Then, we provide the converse construction.

\begin{theorem}[From C-OPGs to OPAs]
Let $(\Sigma, M)$ be an OP-alphabet. For any C-OPG defined thereon an equivalent OPA can be effectively built.\label{AOPG-OPA}
\end{theorem}

\begin{proof} 

A nondeterministic OPA\footnote{Any nondeterministic OPA can be transformed into a deterministic one at the cost of  quadratic exponential increase in the size of the state space \cite{LonatiEtAl2015}.} 
$\mathcal A = \langle  \Sigma, M, Q, I,$ $F,  \delta \rangle$ from a given C-OPG $G$ with the same precedence matrix $M$ as $G$ is built in such a way that a successful computation thereof corresponds to
	building bottom-up a syntax tree of $G$: the automaton performs a push transition when it reads the first terminal of a new rhs;
	it performs a shift transition when it reads a terminal symbol inside a rhs, i.e.\ a leaf with some left sibling leaf.
	It performs a pop transition when it completes the recognition of a rhs, then it guesses (nondeterministically) the nonterminal at the lhs.
	Each state contains two pieces of information: the first component represents the prefix of the rhs under construction,
	whereas the second component is used to recover the rhs \emph{previously under construction} (see Figure~\ref{fig:OPGOPA-trees})
	whenever all rhs nested below have been completed.
  Precisely, the construction of $\mathcal A$ is defined as follows.

 \begin{figure}
\centering		
\begin{minipage}{0.3\textwidth}
\centering
\begin{tikzpicture}[scale=0.5,
triangle/.style={isosceles triangle,draw=,shape border rotate=90,isosceles triangle stretches=true, inner sep=0,yshift={0mm}},
small triangle/.style={triangle, minimum height = 5mm, minimum width = 3mm },
]
\node{$\dots$}
child{node (leftb) {$\dots$}}
child{node{$\dots$}}
child{node (rightb) {$\dots$}}
child {node {$B$}
    child {node{ $\dots$}
        child {node {$A$}
            child{node (lefta) {$\dots$}}
            child{node{$\dots$}}
            child{node (righta) {$\dots$}}
        }
        child {node {}
            edge from parent [path]
        }
    }
    child {node {}
        edge from parent [path]
    }
}
child {node {}
    edge from parent [path]
}
;
\draw[decorate,decoration={brace,mirror,raise=5pt}] (leftb.west) -- (rightb.east) node[pos=.5,below=7pt] {$\beta$};
\draw[decorate,decoration={brace,mirror,raise=5pt}] (lefta.west) -- (righta.east) node[pos=.5,below=7pt] {$\alpha$};
\end{tikzpicture}
\end{minipage}
\begin{minipage}{0.3\textwidth}
\centering
\begin{tikzpicture}[scale=0.5,
triangle/.style={isosceles triangle,draw=,shape border rotate=90,isosceles triangle stretches=true, inner sep=0,yshift={0mm}},
small triangle/.style={triangle, minimum height = 5mm, minimum width = 3mm },
]
\node {$\dots$}
child{node (leftb) {$\dots$}}
child{node{$\dots$}}
child{node (rightb) {$\dots$}}
child {node {$A$}
    child{node (lefta) {$\dots$}}
    child{node{$\dots$}}
    child{node (righta) {$\dots$}}
}
child {node {}
    edge from parent [path]
}
;
\draw[decorate,decoration={brace,mirror,raise=5pt}] (leftb.west) -- (rightb.east) node[pos=.5,below=7pt] {$\beta$};
\draw[decorate,decoration={brace,mirror,raise=5pt}] (lefta.west) -- (righta.east) node[pos=.5,below=7pt] {$\alpha$};
\end{tikzpicture}
\end{minipage}
\caption{
    When parsing $\alpha$, the prefix previously under construction is $\beta$.}
\label{fig:OPGOPA-trees}
\end{figure}
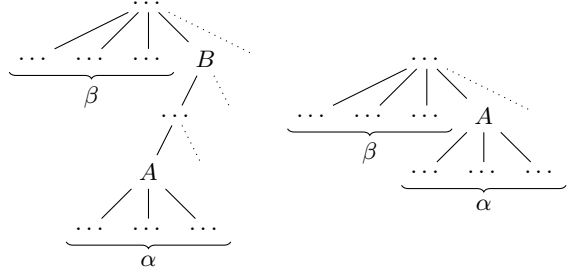

%
Let $\hat P$ be the set of rhs $\gamma$ where all $^+$ and related parentheses have been erased.
Let $\tilde{P}$ be the set of strings $\tilde{\gamma} \in V^+$ belonging to the language of some rhs $\gamma$ of $P$ 
that is inductively defined as follows:\\
if $(\eta)^+$ is a subexpression of $\gamma$ such that $\eta$ is a single string $\in V^+$ then $\tilde {\eta}= \{\eta, \eta\eta\}$;\\
if $\eta = \alpha_1 (\beta_1)^+ \alpha_2 (\beta_2)^+ \dots \alpha_n$ where $\alpha_i \in V^*$,
then $\tilde{\eta} = \{\eta_1, \eta_1\eta_1\}$
where $\eta_1 = \alpha_1 \tilde\beta_{1} \alpha_2 \tilde\beta_{2} \dots \alpha_n$.

 E.g., let $\eta$ be $(Ba(bc)^+)^+$; then $\tilde \eta = \{Babc, Babcbc, BabcBabc, BabcbcBabc,\\ BabcBabcbc, BabcbcBabcbc \}$
 and $ \hat \eta = \{Babc\}$
 
 Let
$
\pref = \{ \alpha \in V^* \Sigma \mid \exists A \to \eta \in P \land 
\exists \alpha, \beta \mid \alpha \beta \in \tilde\eta\}
$
 be the set of prefixes, ending with a terminal symbol, of strings $\in \tilde{P}$;
	define $\mathbb Q = \{ \varepsilon \} \cup \pref \cup N$, $Q = \mathbb Q \times (\{ \varepsilon \} \cup \pref)$,
	$I = \{ \pair{\varepsilon}{\varepsilon} \}$,
	and $F = S \times \{\varepsilon\} \cup \{ \pair{\varepsilon}{\varepsilon} \text{  if  } \varepsilon \in L(G)\}$.
	Note that $|\mathbb{Q}| = 1 + |\pref| + |N|$ is $O(m^h)$ where $m$ is the maximum length of the rhs in $P$, and $h$ is the maximum nesting level of $^+$ operators in rhs; therefore $|Q|$ is $O(m^{2h})$.
	
The transition functions are defined by the following formulas, for
$a \in \Sigma$ and $\alpha, \alpha_1, \alpha_2 \in \mathbb Q$, $\beta, \beta_1, \beta_2 \in \{ \varepsilon \} \cup \pref$, and where for any expression $\xi$, $\bar\xi$ is obtained from $\xi$ by erasing parentheses and $^+$ operators:
\begin{itemize}
\item \(
	\delta_{\text{shift}}( \pair{\alpha}{\beta}, a) \ni 
	\begin{cases}
        \text{if } \alpha \not\in N:
        &
        \begin{cases}
            \pair{\bar{\eta} \bar{\zeta}}{\beta} &
            \text{if}
            \left(
            \begin{array}{l}
            \exists A \to \gamma \mid
            \gamma = \eta (\zeta)^+ \theta \land\\
            \alpha a = \bar{\eta} \bar{\zeta} \bar{\zeta},  \alpha a \bar{\theta} \in L(\gamma) \cap \tilde P
            \end{array}
            \right) \\
            \pair{\alpha a}{\beta} & \text{otherwise}
        \end{cases}
        \\
        \text{if } \alpha \in N:
        &
        \begin{cases}
            \pair{\bar{\eta} \bar{\zeta}}{\beta} &
            \text{if}
            \left(
            \begin{array}{l}
            \exists A \to \gamma \mid
            \gamma = \eta (\zeta)^+ \theta \land\\
            \beta \alpha a = \bar{\eta} \bar{\zeta} \bar{\zeta},  \alpha a \bar{\theta} \in L(\gamma) \cap \tilde P
            \end{array}
            \right) \\
            \pair{\beta \alpha a}{\beta} & \text{otherwise}
        \end{cases}
        \\
    \end{cases}
\)


  
  

\item
    $\delta_{\text{push}} ( \pair{\alpha}{\beta}, a) \ni
    \left\{\begin{array}{ll}
    \pair{a}{\alpha} & \text{if } \alpha \not\in N\\
    \pair {\alpha a}{\beta} & \text{if } \alpha \in N
    \end{array}\right.$
		%
\item 
	$\delta_{\text{pop}}(\pair{\alpha_1}{\beta_1}, \pair{\alpha_2}{\beta_2}) \ni \pair{A}{\gamma}$
\\
    for every $A$ such that
	$\left[\begin{array}{ll}
		\text{if } \alpha_1 \notin N: 
            A \to \alpha \in P \land \alpha_1 \in L(\alpha) \cap \hat{P}
    & \\
        \text{if } \alpha_1 \in N:
		      A \to \delta \in P \land \beta_1 \alpha_1 \in L(\delta) \cap \hat P
    \end{array}\right.$
\\
	and $\gamma = \left\{\begin{array}{ll}
	\alpha_2 & \text{if } \alpha_2 \notin N\\
	\beta_2 & \text{if } \alpha_2 \in N.
	\end{array}\right.$
\end{itemize}

The states reached by push and shift transitions have the first component in $\mathbb P$. 
If state $\pair{\alpha}{\beta}$ is reached after a push transition, then
$\alpha$ is the prefix of the rhs (deprived of the $^+$ operators) that is currently under construction and
$\beta$ is the prefix previously under construction;
in this case $\alpha$ is either a terminal symbol or a nonterminal followed by a terminal one. 
	
If the state is reached after a shift transition,
and the $\alpha$ component of the previous state was not a single nonterminal,
then the new $\alpha$ is the concatenation of the first component of the previous state with the read character.
If, instead, the $\alpha$ component of the previous state was a single nonterninal
---which was produced by a pop transition---
then the new $\alpha$ also includes the previous $\beta$ (see Figure \ref{fig:OPGOPA-trees}) and $\beta$ is not changed from the previous state. 
However, if the new $\alpha$ becomes such that a suffix thereof is a double occurrence of a string $\zeta \in L((\zeta)^+)$ ---hence $\alpha \in \pref$--- then the second occurrence of $\zeta$ is cut from the new $\alpha$, which therefore becomes a prefix of an element of $\hat{P}$.


The states reached by a pop transition have the first component in $N$: 
if $\pair{A}{\gamma}$ is such a state, then $A$ is the corresponding lhs, and $\gamma$ is the prefix previously under construction.

For instance, imagine that a C-OPG contains the rules $A \to (Ba(bc)^+)^+ a$ and $B \to h$ 
and that the corresponding OPA $\mathcal A$ parses the string $habcbchabca$: 
after scanning the prefix $habcb$ $\mathcal A$ has reduced $h$ to $B$ and has $Babcb$ as the first component of its state; 
after reading the new $c$ it recognizes that the suffix of the first state component would become a second instance of $bc$ belonging to $(bc)^+$; 
thus, it goes back to $Babc$. 
Then, it proceeds with a new reduction of $h$ to $B$ and, 
when reading with a shift the second $a$ appends $Ba$ to its current $\beta$ which was produced by the previous pop so that the new $\alpha$ becomes $BabcBa$; 
after shifting $b$ it reads $c$ and realizes that its new $\alpha$ would become $BabcBabc$, i.e., an element of $(Ba(bc)^+)^+$ and therefore ``cuts'' it to the single instance thereof, i.e., $Babc$. Finally, after having shifted the last $a$ it is ready for the last pop.
 
Notice that the result of $\delta_{\text{shift}}$ and $\delta_{\text{push}}$ is a singleton, 
whereas $\delta_{\text{pop}}$ may produce several states, in case of repeated rhs.
Thus, if $G$ is BD, the corresponding $\mathcal A$ is deterministic.

The equivalence between $G$ and $\mathcal A$ derives from the following Lemmas~\ref{lemma:compToDer} and \ref{lemma:derToComp}, 
when
$\beta=\gamma=\varepsilon$, $\Pi = \bot$ and $A$ is an axiom.
Their statements are identical to Lemmas 3.2 and 3.3 given in \cite{LonatiEtAl2015} for the original construction applied to non-C-OPGs.
\end{proof}

\begin{lemma}
		\label{lemma:compToDer}
		Let $x$ be the body of a chain
		and $\beta, \gamma \in \mathbb P\cup \{\varepsilon\}$.
		Then 
		$\pair{\beta}{\gamma}  \ourpath x q$
		implies the existence of $A \in N$ such that  $A \stackrel \ast \Rightarrow x \ \text{ in } G$ and $q = \pair A {\beta}$.
	\end{lemma}
 \begin{lemma}
		\label{lemma:derToComp}
		Let $x$ be the body of a chain and $A \in N$.
		Then, 
		$
		A \stackrel \ast \Rightarrow x \ \text{ in } G 
		$
		implies 
		$
		\pair{\beta}{\gamma}  
		\ourpath x {\pair A {\beta}} 
		$
		for every $\beta, \gamma \in \mathbb P \cup \{\varepsilon\}$.
	\end{lemma}

The proof of the above lemmas is based on a natural induction on the depth $h$ of the chains visited by the OPA and is not repeated here. 
We simply point out the only key difference due to the fact that simple chains generated by C-OPGs can have unbounded length.

During a series of shift transitions scanning a simple chain the $\beta$ component of the state remains unchanged,
whereas the $\alpha$ component at each transition appends the read character ---always under the constraint that it belongs to $\pref$---
up to the point where the last append would repeat a suffix of  $\alpha$ ---say a string $y$ in case of simple chains--- identically; 
if $y$ occurs under the $^+$ operator of a rhs of which $\alpha$ is  a prefix then the second occurrence of $y$ is 
erased 
and the state is the same as after reading its first occurrence, 
thus closing a loop in the same way as it happens with finite state machines.
As a consequence the OPA can repeat the loop reading $y$ any number of times just as the C-OPG can generate, in an immediate derivation, any number of occurrences of $y$ thanks to the $^+$ operator.

The converse statement that if $A \xLongrightarrow[]{} x$ then $
		\pair{\beta}{\gamma}  
		\ourpath x {\pair A {\beta}} 
		$
		for every $\beta, \gamma \in \mathbb P \cup \{\varepsilon\}$
follows directly from the construction of $\mathcal A$.

The generalization of the above reasoning to the case of composed chains perfectly parallels the induction step of the corresponding lemmas of \cite{LonatiEtAl2015}.

\begin{example}
Figure \ref{fig:run cyclic} displays a run of the OPA obtained from the C-OPG of Figure \ref{fig:cyclic AE} accepting the sentence $n+n+n/n/n+n+n$.

\begin{figure}[bt]
\centering
\begin{scriptsize}
\[
		\begin{array}{|l|c|r|}
		\hline
		\text{stack}   											 
		&  \text{state}
		&  \text{current input}\\\hline
  \bot & \pair{\varepsilon} {\varepsilon}  &  n + n + n / n /  n + n +n \# \\ \hline 
  \bot \tstack{n} {\pair{\varepsilon} {\varepsilon}}   &  \pair{n} {\varepsilon}
   &  + n + n / n /  n + n +n \# \\ \hline 
   \bot   &  \pair{T} {\varepsilon}
   &  + n + n / n /  n + n +n \# \\ \hline 
   \bot \tstack{+} {\pair{T} {\varepsilon}}   &  \pair{T+} {\varepsilon}
   &  n + n / n /  n + n +n \# \\\hline 
  \bot \tstack{+} {\pair{T} {\varepsilon}}
  \tstack{n} {\pair{T+} {\varepsilon}}
  &  \pair{n} {T+}
   &  + n / n /  n + n +n \# \\\hline 
   \bot \tstack{+} {\pair{T} {\varepsilon}}
  &  \pair{T} {T+}
   &  + n / n /  n + n +n \# \\\hline 
    \bot \tstack{+} {\pair{T} {\varepsilon}}
  &  \bf{\pair{T+} {T+}}
   &  n / n /  n + n +n \# \\\hline 
     \bot \tstack{+} {\pair{T} {\varepsilon}}
     \tstack{n} {\pair{T+} {T+}}
  &  \pair{n} {T+}
   &  / n /  n + n +n \# \\\hline 
    \bot \tstack{+} {\pair{T} {\varepsilon}}
  &  \pair{D} {T+}
   &  / n /  n + n +n \# \\\hline 
    \bot \tstack{+} {\pair{T} {\varepsilon}}
    \tstack{/} {\pair{D} {T+}}
  &  \pair{D/} {T+}
   &  n /  n + n +n \# \\\hline 
    \bot \tstack{+} {\pair{T} {\varepsilon}}
    \tstack{/} {\pair{D} {T+}}
    \tstack{n} {\pair{/} {D}}
  &  \pair{n} {D/}
   &  /  n + n +n \# \\\hline 
   \bot \tstack{+} {\pair{T} {\varepsilon}}
    \tstack{/} {\pair{D} {T+}}
  &  \pair{E} {D/}
   &  /  n + n +n \# \\\hline
    \bot \tstack{+} {\pair{T} {\varepsilon}}
  &  \pair{D} {T+}
   &  /  n + n +n \# \\\hline
     \bot \tstack{+} {\pair{T} {\varepsilon}}
     \tstack{/} {\pair{D} {T+}}
  &  \pair{D/} {T+}
   &  n + n +n \# \\\hline
     \bot \tstack{+} {\pair{T} {\varepsilon}}
     \tstack{/} {\pair{D} {T+}}
     \tstack{n} {\pair{D/} {T+}}
  &  \pair{n} {D/}
   &  + n +n \# \\\hline
   \bot \tstack{+} {\pair{T} {\varepsilon}}
     \tstack{/} {\pair{D} {T+}}
  &  \pair{E} {D/}
   &  + n +n \# \\\hline
    \bot \tstack{+} {\pair{T} {\varepsilon}}
  &  \pair{T} {T+}
   &  + n +n \# \\\hline
    \bot \tstack{+} {\pair{T} {\varepsilon}}
  &  \bf{\pair{T+} {T+}}
   & n +n \# \\\hline
    \bot \tstack{+} {\pair{T} {\varepsilon}}
    \tstack{n} {\pair{T+} {T+}}
  &  \pair{n} {T+}
   & +n \# \\\hline
    \bot \tstack{+} {\pair{T} {\varepsilon}}
  &  \pair{T} {T+}
   & +n \# \\\hline
    \bot \tstack{+} {\pair{T} {\varepsilon}}
  &  \bf{\pair{T+} {T+}}
   & n \# \\\hline
     \bot \tstack{+} {\pair{T} {\varepsilon}}
      \tstack{n} {\pair{T+} {T+}}
  &  \pair{n} {T+}
   &  \# \\\hline
     \bot \tstack{+} {\pair{T} {\varepsilon}}
  &  \pair{T} {T+}
   &  \# \\\hline
    \bot 
  &  \pair{P} {\varepsilon}
   &  \# \\\hline
    \end{array}
 \]
\end{scriptsize}
\caption{A run of the OPA built from the C-OPG of Figure \ref{fig:cyclic AE} accepting the sentence $n+n+n/n/n+n+n$.
The states truncated by erasing a repeated suffix $\zeta$ occurring under the scope of a $^+$ operator are emphasized in boldface.}
\label{fig:run cyclic}
\end{figure}    
\end{example}\label{ex:sample computation}

The construction of a C-OPG equivalent to a given OPA is far simpler than the converse one,  thanks to the explicit structure associated to words by the precedence matrix. 
The key difference w.r.t. the analogous construction given in \cite{LonatiEtAl2015} is that even simple chains can have unbounded length due to the fact that the $\doteq $ relation may be circular.

First, in analogy with the definition of $\tilde P$, we define \emph{essential supports} for both simple and composed chains, 
as those supports where possible cyclic behaviors of the OPA along a sequence of terminals ---whether with interposing nonterminals or not--- occur exactly twice.

\begin{definition}\label{def:essential supports}
(Essential chain supports)
An essential support of a simple chain $\chain {a_0} {a_1 a_2 \dots a_n} {a_{n+1}}$
	is any path in $\mathcal A$ of the form
	\begin{equation}
	q_0
	\va{a_1}{q_1}
	\shift{}{}
	\dots
	\shift{}q_{n-1}
	\shift{a_{n}}{q_n}
	\flush{q_0} {q_{n+1}}
	\end{equation}

where any cycle $q_i a_{i1} ...a_{ik}  q_i $ is repeated exactly twice.

Essential supports for composed chains are defined similarly.
\end{definition}

For instance, with reference to the OPA built from the C-OPG of Figure \ref{fig:cyclic AE} an essential support of the chain 
$\chain {\#} {n+n+n+n+n} {\#}$ 
is:
\[
\begin{array}{l}
	\pair{\varepsilon} {\varepsilon}
 \ourpath{n}{\pair{T} {\varepsilon}}
	\va{+}\pair{T+} {\varepsilon}
	 \ourpath{n}{\pair{T} {T+}}
  \shift{+} {\pair{T+} {T+}}
   \ourpath{n}{\pair{T} {T+}}
  \shift{+} \\ 
  {\pair{T+} {T+}}
   \ourpath{n}{\pair{T} {T+}} \flush{\pair{T} {\varepsilon}} {\pair{P} {\varepsilon}}
   \end{array}
   \]

\begin{lemma}\label{lemma:essential support}
The essential supports of simple chains of any OPA have an effectively computable bounded length.
\end{lemma}

\begin{theorem}[From OPAs to C-OPGs]
Let $(\Sigma, M)$ be an OP-alphabet. For any OPA defined thereon an equivalent C-OPG can be effectively built.\label{OPA-AOPG}
\end{theorem}

\begin{proof}
		Given an \opa \ $\mathcal A = \langle \Sigma,M, Q, I, F, \delta \rangle $, 
		we show how to build an equivalent \opg \ $G$ having operator precedence matrix M. The equivalence between $\mathcal A$ and $G$ should then be rather obvious.
		
		$G$'s nonterminals are the 4-tuples $(a, q, p, b) \in \Sigma \times Q \times Q \times \Sigma$, written as $\nont apqb$.
		$G$'s rules are built as follows:
		\begin{itemize}
			\item 
			for every essential support 
			of a simple chain, the rule 
			\[
			\nont {a_0}{q_0}{q_{n+1}}{a_{n+1}}\longrightarrow a_1 a_2 \dots a_n\ ;
			\]
			
   where every double sequence 
   $a_{i1} ...a_{ik} a_{i1}...a_{ik}$ is recursively replaced by $(a_{i1} ...a_{ik})^+$ 
   by proceeding from the innermost cycles to the outermost ones,
   is in $P$; 
   furthermore, if $a_0 = a_{n+1} = \#$, $q_0$ is initial, and $q_{n+1}$ is final, then
			$\nont {\#}{q_0}{q_{n+1}}{\#}$ is in $S$;
			\\
			\item 
			for every essential support of a composed chain,
			add the rule 
			\[
			\nont {a_0}{q_0}{q_{n+1}}{a_{n+1}} \longrightarrow \Lambda_0 a_1 \Lambda_1 a_2 \dots a_n \Lambda_n \ ;
			\]
			where, for every $i = 0,1, \dots, n$,
			$\Lambda_i = \nont {a_i}{q_i}{q'_i}{a_{i+1}}$ if $x_i \neq \varepsilon$ and $\Lambda_i =
			\varepsilon$ otherwise, by replacing double cyclic sequences $\alpha_i \alpha_i$ with $(\alpha_i)^+ $ in the same way as for simple chains;  furthermore, if $a_0 = a_{n+1} = \#$, $q_0$ is initial, and $q_{n+1}$ is final, then add
			$\nont {\#}{q_0}{q_{n+1}}{\#}$ to $S$,
			and, if $\varepsilon$ is accepted by $\mathcal{A}$, add $A \to \varepsilon$, $A$ being a new axiom not otherwise occurring in any other rule.
		\end{itemize}

		Notice that the above construction is effective thanks to Lemma \ref{lemma:essential support} and to the fact that subchains of composed chains are replaced by nonterminals $\Lambda_i$.
		\hfill\end{proof}

  \noindent \emph{Remark.} 
  The definition of C-OPG and the constructions of Theorems \ref{AOPG-OPA} and \ref{OPA-AOPG} have been given with the same approach as used in \cite{LonatiEtAl2015}, i.e., avoiding possible optimizations to keep technical details as simple and essential as possible. 
  For instance, we allowed only the $^+$ operator in grammar's rhs as the minimum necessary extension to obtain the expressive power of OPAs: 
  allowing, e.g., to use set union within the scope of a $^+$ operator would have allowed in some cases more compact grammars but such a choice would have also caused an explosion of the cardinality of the $\tilde{P}$ set. 
  For the same reason we excluded a priori renaming and $\varepsilon$-rules in our grammars.

\section{Other equivalences among OPL formalisms}\label{equivalences}
\begin{figure}
\begin{tabular}{m{0.27\textwidth}m{0.71\textwidth}}
  \centering{
   \includegraphics[width=0.25\textwidth]{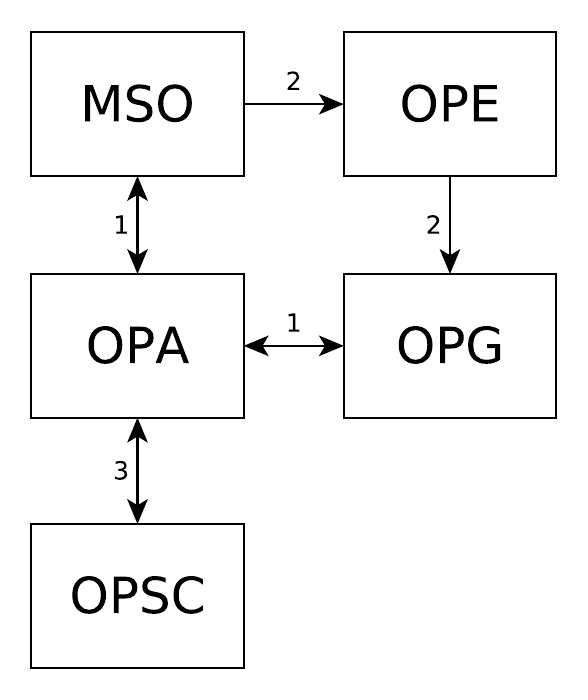}}
&
 \begin{flushleft}
   \begin{scriptsize}
 {\bf Legend}\\
  All boxes denote classes of OPLs with a common OPM:
  \begin{itemize}
      \item 
  MSO denotes OPLs defined through MSO formulas
      \item 
  OPA denotes OPLs defined through OP automata \cite{LonatiEtAl2015}
      \item 
        OPSC denotes OPLs defined through
        a Syntactic Congruence with a finite number of equivalence classes \cite{Henzinger23}
   \item 
  $\text{OPE}$ denotes OPLs defined through OPEs \cite{MPC23}
  \end{itemize}
  Arrows between boxes denote language family inclusion; they are labeled by the reference pointing 
  to paper where the property has been proved 
  (1 is \cite{LonatiEtAl2015}, 2 is \cite{MPC23}, 3 is \cite{Henzinger23}).
\end{scriptsize}
\end{flushleft}
\end{tabular}
\caption{The relations among the various characterizations of OPLs. }
\label{fig:figurona}
\end{figure}
Figure~\ref{fig:figurona} 
displays the five equivalent formalisms introduced in the literature to define OPLs. 
The reciprocal inclusions between MSO and OPA, between OPA and the finite equivalence classes of OPL syntactic congruence, and the inclusion of MSO in OPE have been proved without the hypothesis of non-circularity of the $\doteq$ relation; 
the reciprocal inclusions between C-OPG and OPA have been restated in Section \ref{A-OPG vs OPA};
it remains to consider the inclusion of OPE in OPG which was proved in \cite{MPC20} under the restrictive hypothesis.
The proof used the claim that deriving an OPG from an OPE may exploit the closure properties of OPLs, in particular, w.r.t. the Kleene $^*$ operator; 
such a closure, however, was proved in \cite{Crespi-ReghizziM12} by using OPGs, again, under the restrictive hypothesis.

Although we will see, in the next section, that all closure properties of OPLs still hold even when there are circularities in the $\doteq $ relation, 
here it is convenient to consider the case of the Kleene $^*$ operator of OPEs separately.
The reason is that in OPEs the Kleene $^*$ operator is now applied to subexpressions independently on the OP relations between their last and first terminal character.

Thus, it is first convenient to rewrite the OPE in a normal form using the $^+$ operator instead of the $^*$ one to avoid having to deal explicitly with the case of the $\varepsilon$ string.
Then, subexpressions of type $(\alpha)^+$ where the last terminal of $\alpha$ is not in relation $\doteq$ with the first one are replaced by the same procedure defined in \cite{Crespi-ReghizziM12} to prove effectively the closure w.r.t. the $^*$ operator.
The new rules will produce a right or left-linear subtree of the occurrences of $\alpha$ depending on the OP relation between the two extreme 
terminals of $\alpha$ and will avoid the use of the $^*$ and $^+$ operators which are not permitted in the original OPGs.

The remaining substrings including the $^+$ operator are the new rhs of the C-OPG.
The other technicalities of the construction of an OPG equivalent to an OPE are identical to those given in \cite{MPC20} and are not repeated here.

\section{OPLs closure properties revisited}\label{closures}

All major closure properties of OPLs have been originally proved by referring to their generating OPGs and some of them, in particular the closure w.r.t. the Kleene$^*$ operator, required the $\doteq$-acyclicity hypothesis.
Thus, it is necessary to prove them again. However, since some of those proofs are  technically rather involved, here we simply observe that it is easier to restate the same properties by exploiting OPAs which are now fully equivalent to C-OPGs.

Observe, in fact, that, thanks to the determinization of nondeterministic OPAs, closure w.r.t. boolean operations is ``for free''.
Closure w.r.t. concatenation can be seen as a corollary of the closure proved in \cite{LonatiEtAl2015} of the concatenation  between an OPL whose strings have finite length and an $\omega$-OPL, i.e., a language of infinite strings.
The construction is based on a nondeterministic guess of the position where a string of the first language could end ---and a nontrivial technique to decide whether it could accepted even in the absence of the $\#$ delimiter---.
Then, the closure w.r.t. Kleene $^*$ is obtained simply by allowing the OPA to repeat such a guess any number of times until the real $^\#$ is found.

\section{Conclusion}\label{concl}
We have filled up a longstanding ``hole'' in the theory of OPLs under the pressure of recent practical applications in the field of parallel compilation that showed how such a hole could hamper the benefits of parallelism \cite{DBLP:conf/hpcasia/LiT23}. 
The new grammatical formalism of C-OPGs, now fully equivalent to OPAs, MSO-logic, and OPEs in the definition of OPLs,
can therefore be exploited to revisit the parallel compilation techniques of \cite{DBLP:conf/hpcasia/LiT23}
(a rather unusual case where theory comes \emph{after} its motivating application) 
or to improve the efficiency of techniques based on the less powerful OPGs  \cite{BarenghiEtAl2015}.

Other algebraic and logic properties of OPLs, e.g., aperiodicity, star-freeness, first-order definability \cite{McNaughtPap71,MPC20} can be re-investigated at the light of this generalization.

\bibliographystyle{splncs04}
\bibliography{opbib}
\end{document}